\documentclass[12pt,a4paper]{article}   
\usepackage[utf8]{inputenc}
\usepackage[english]{babel}

\usepackage[theorems, global-numbering, colorlinks]{default}

\usepackage[left=2cm,right=2cm,top=2.5cm,bottom=2cm,headheight=15pt]{geometry}

\usepackage{biblatex}
\usepackage{hyperref}
\usepackage{xurl}
\hypersetup{breaklinks=true}
\title{Alternative proof for the bias\\ of the hot hand statistic of streak length one}
\author{Maximilian Janisch\footnote{Institut für Mathematik, Universität Zürich. E-Mail: \url{maximilian.janisch@math.uzh.ch} or\\\url{mail@maximilianjanisch.com}.}}
\date{July 15, 2024}

\begin{document}
	\maketitle
    \thispagestyle{empty}

    \begin{abstract}
        For a sequence of $n$ random variables taking values $0$ or $1$, the \emph{hot hand statistic of streak length $k$} counts what fraction of the streaks of length $k$, that is, $k$ consecutive variables taking the value $1$, among the $n$ variables are followed by another $1$. Since this statistic does not use the expected value of how many streaks of length $k$ are observed, but instead uses the realization of the number of streaks present in the data, it may be a biased estimator of the conditional probability of a fixed random variable taking value $1$ if it is preceded by a streak of length $k$, as was first studied and observed explicitly in [Miller and Sanjurjo, 2018]. In this short note, we suggest an alternative proof for an explicit formula of the expectation of the hot hand statistic for the case of streak length one. This formula was obtained through a different argument in [Miller and Sanjurjo, 2018] and [Rinott and Bar-Hillel, 2015].
    \end{abstract}

    For $n\in\N$, let $(X_j)_{j\in\set{1,\dots, n}}$ be a sequence of independent, identically distributed random variables such that $\P(X_1=1)=1-\P(X_1=0)=p\in[0,1]$. We are interested in testing for \emph{hot streaks}, that is, we want to design a test statistic which infers, from a sequence of zeros and ones, whether the sequence of successes is independent or whether there is dependence among the successes.

    We consider here the test statistic studied in \cite{Miller-Sanjurjo} (see also \cite{rinott2015comments}), which will be referred to as the \emph{hot hand statistic}. The hot hand statistic is obtained by counting how often a success follows a streak of $k\in\N$ successes.

    \begin{definition}[Hot hand statistic]
        For $n\in\N_{\ge 2}$ and any family of random variables \\ $X=(X_j)_{j\in\set{1,\dots,n}}$ taking values in $\set{0,1}^n$, the \emph{hot hand statistic for streak length $k\in\N_{\le n-1}$} is defined as 
        \begin{equation}\label{eq:Pkhat}
            \hat P_k(X)\define \frac{X_1 X_2 \dots X_{k+1}+\cdots+X_{n-k}X_{n-k+1}\dots X_n}{X_1X_2\dots X_k+X_2 X_3\dots X_{k+1}+\cdots+X_{n-k}X_{n-k+1}\dots X_{n-1}}.
        \end{equation}
        Note that $\hat P_k(X)$ is only well-defined for those realizations of $X$ where the denominator of \eqref{eq:Pkhat}, henceforth denoted by $D_k(X)$, is non-zero.
    \end{definition}

    We will be specifically interested in the case $k=1$, which we write out here for clarity of exposition:
    \begin{equation*}
        \hat P_1(X) = \frac{X_1 X_2 + X_2 X_3 + \dots + X_{n-1} X_n}{X_1+ X_2+\dots+X_{n-1}}.
    \end{equation*}

    In \cite[Theorem 1]{Miller-Sanjurjo} it was established that, if $n\in\N_{\ge 3}$, $k\in\N_{\le n-2}$, and $X_1,\dots, X_n$ is a family of independent, identically distributed random variables with \begin{equation}\label{eq:Bernoulli}
    \P(X_1=1)=1-\P(X_1=0)=p\in(0,1),
    \end{equation}
    then
    \begin{equation}\label{eq:bias}
        \E\left(\hat P_k(X)\mid D_k(X)\neq 0\right) < p.
    \end{equation}

    \begin{remark}[Context for {\cite[Theorem 1]{Miller-Sanjurjo}}]
        The preceding result may be seen in rough analogy to how the naive sample variance estimator is biased as it uses an empirical estimate of the expected value of the underlying distribution instead of the actual expected value. It is furthermore in line with a broader literature on estimation of Markov chain transition probabilities, see e.g. \cite{Bai}.
    \end{remark}

    Furthermore, the following explicit formula\footnote{Or rather, a variant which is equivalent to the simplified version presented here.} for the case $k=1$ was established, using an auxiliary random variable $\tau$ which, conditionally on $X=(X_1,\dots,X_n)$, is uniformly distributed on those indices $j\in\{2,\dots,n\}$ for which $X_{j-1}=1$. This argument is formulated fully in \cite[Theorem A.3.2]{Miller-Sanjurjo}, but a shortened and intuitively accessible, albeit mathematically not fully written out, version of the proof can be found on \cite[Page 3]{rinott2015comments}.

    \begin{theorem}[Explicit formula for $k=1$, see {\cite[Theorem A.3.2]{Miller-Sanjurjo}} and {\cite[Page 3]{rinott2015comments}}]\label{thm:explicit}
        Let $n\in\N_{\ge 3}$ and let $X_1,\dots,X_n$ as well as $p$ be as in \eqref{eq:Bernoulli}. Then
        \begin{equation}\label{eq:P1hat expectation}
            \E\left(\hat P_1(X)\mid D_1(X)\neq 0\right) = \frac{p}{1-(1-p)^{n-1}}+ \frac{p-1}{n-1}.
        \end{equation}
    \end{theorem}

    In this short note we aim to provide an alternative proof of Theorem \ref{thm:explicit}, which we deem more direct than the one found in \cite{Miller-Sanjurjo,rinott2015comments}. We start with the following Lemma.

    \begin{lemma}[Expectations of reciprocal of sum of constant and binomially distributed random variable]\label{lem:expectations}
        Let $p\in(0,1]$ and let $n\in\N$. Let $Z_{n,p}$ be a random variable with distribution $\mathrm{Binomial}(n, p)$. Then
        \begin{equation}\label{eq:expectation of 1/(1+Z)}
            \E\left(\frac{1}{1+Z_{n,p}}\right) = \frac{1 - (1-p)^{n+1}}{(n+1)p},
        \end{equation}
        and 
        \begin{equation}\label{eq:expectation of 1/(2+Z)}
            \E\left(\frac{1}{2+Z_{n,p}}\right)=\frac{(1-p)^{n+2}+(n+2) p-1}{(n+1) (n+2) p^2}
        \end{equation}
    \end{lemma}

    \begin{proof}[Proof of Lemma \ref{lem:expectations}]
        First, for $z\in(-1,\infty)$, $\frac{1}{1+z} = \int_0^1 t^z\,\mathrm dt$. Therefore, using Fubini-Tonelli and recognizing $\E\left(t^{Z_{n,p}}\right)$ as the probability generating function of a binomial distribution, we get
        \begin{equation*}
            \E\left(\frac{1}{1+Z_{n,p}}\right) = \E\left(\int_0^1 t^{Z_{n,p}} \,\mathrm dt\right) = \int_0^1 \E\left(t^{Z_{n,p}}\right)\,\mathrm dt = \int_0^1 ((1-p) + p t)^n\,\mathrm dt.
        \end{equation*}
        By making the change of variables $s = (1-p) + pt$, one obtains that the latter integral equals the right-hand side of \eqref{eq:expectation of 1/(1+Z)}, thus proving \eqref{eq:expectation of 1/(1+Z)}.

        Second, for $n\in\N$, noting that $Z_{n+1,p}$ equals in distribution a $\mathrm{Bernoulli}(p)$-distributed random variable, independent of $Z_{n,p}$, plus $Z_{n,p}$, we get 
        \begin{equation*}
            \E\left(\frac{1}{1+Z_{n+1,p}}\right) = p \E\left(\frac{1}{2+Z_{n,p}}\right) + (1-p) \E\left(\frac{1}{1+Z_{n,p}}\right).
        \end{equation*}
        Rearranging and using \eqref{eq:expectation of 1/(1+Z)} yields \eqref{eq:expectation of 1/(2+Z)}.
    \end{proof}
    We are now in a position to prove Theorem \ref{thm:explicit}.
    \begin{proof}[Proof of Theorem \ref{thm:explicit}]
        Firstly, by conditioning on $X_{n-1}=X_n=1$, we obtain
        \begin{equation}\begin{split}
            &\E\left(\frac{X_{n-1}X_n}{X_1+\dots+X_{n-1}}\ \middle\vert\ D_1(X)\neq 0\right)\\ &= \E\left(\frac{1}{1+X_1+X_2+\dots+X_{n-2}}\right)\P\left(X_{n-1}=X_n=1\mid X_1+\dots+X_{n-1}>0\right) \\
            &= \E\left(\frac{1}{1+Z_{n-2,p}}\right) \frac{p^2}{1-(1-p)^{n-1}}.
        \end{split}\end{equation}
        From \eqref{eq:expectation of 1/(1+Z)} we thus get
        \begin{equation*}
            \E\left(\frac{X_{n-1}X_n}{X_1+\dots+X_{n-1}}\ \middle\vert\ D_1(X)\neq 0\right) = \frac{p}{n-1}.
        \end{equation*}
        Secondly, for $j\in\set{2,\dots, n-1}$, we have, by exchangeability of the $X_j$,
        \begin{equation}\label{eq:other-indices}\begin{split}
            \E\left(\frac{X_{j-1}X_j}{X_1+\dots+X_{n-1}}\ \middle\vert\ D_1(X)\neq 0\right) &= \E\left(\frac{X_1 X_2}{X_1+\dots+X_{n-1}}\ \middle\vert\ D_1(X)\neq 0\right) \\ 
            &= \E\left(\frac{1}{2+Z_{n-3,p}}\right)\frac{p^2}{1-(1-p)^{n-1}} \\
            &=\frac{\frac{p}{1-(1-p)^{n-1}}-\frac{1}{n-1}}{n-2},
        \end{split}\end{equation}
        where we have used \eqref{eq:expectation of 1/(2+Z)} for the last identity. We thus obtain
        \begin{equation*}
            \E(\hat P_1(X)\vert D_1(X)\neq 0) =\sum_{j=2}^n \E\left(\frac{X_{j-1}X_j}{X_1+\dots+X_{n-1}}\ \middle\vert\ D_1(X)\neq 0\right) =\frac{p}{1-(1-p)^{n-1}}+ \frac{p-1}{n-1},
        \end{equation*}
        thus establishing \eqref{eq:P1hat expectation}.
    \end{proof}

    \begin{remark}[Theorem \ref{thm:explicit} implies \eqref{eq:bias} for the case $k=1$]
        We note by full expansion and using $p\in(0,1)$ that
        \begin{equation*}\begin{split}
            &\frac{p}{1-(1-p)^{n-1}} + \frac{p-1}{n-1}<p \\ &\iff p (n-1) + (p-1) (1-(1-p)^{n-1}) < p (n-1) (1-(1-p)^{n-1}) \\
            &\iff (n-1) (1-p)^{n-2} < 1 +(n-2) (1-p)^{n-1}.
        \end{split}\end{equation*}
        The latter inequality, however, follows from the weighted AM-GM inequality, as 
        \begin{equation*}
            1+(n-2) (1-p)^{n-1}\ge (n-1) \sqrt[n-1]{((1-p)^{n-1})^{n-2}}=(n-1)(1-p)^{n-2}.
        \end{equation*} 
        Since $(1-p)^{n-1}\neq 1$, this inequality is strict.
        Therefore, Theorem \ref{thm:explicit} implies as Corollary \eqref{eq:bias} when $k=1$.
    \end{remark}

    \begin{remark}[The case $k>1$]
        The presented approach does not seem to generalize well to the case $k>1$, as it would necessitate computing quantities such as 
        \begin{equation*}
            \E\left(\frac{1}{1+X_1\dots X_k + X_2 \dots X_{k+1}+\cdots + X_{n-k}\dots X_{n-1}}\right).
        \end{equation*}
    \end{remark}

    \printbibliography

\end{document}